\documentclass[11pt]{article}
\usepackage{amssymb,amsmath,amsthm,amscd,latexsym}
\usepackage{mathrsfs}
\usepackage{mathrsfs}
\usepackage{amsfonts}
\usepackage{amssymb}
\usepackage[all]{xy}
\usepackage{makecell}
\usepackage{xypic}
\usepackage{cite}
\usepackage{graphicx}
\usepackage{epstopdf}
\usepackage{booktabs}
\usepackage{geometry}
\usepackage{microtype}
\usepackage[justification=centering]{caption}
\usepackage{array, caption, threeparttable}
\usepackage[font=small,labelfont=bf,labelsep=none]{caption}

\makeatletter
\long\def\@makecaption#1#2{%
  \vskip\abovecaptionskip
   \sbox\@tempboxa{#1:#2}%
  \ifdim \wd\@tempboxa >\hsize
   {\bfseries #1:} #2\par
  \else
    \global \@minipagefalse
   \hb@xt@\hsize{\box\@tempboxa\hfil}%
  \fi
  \vskip\belowcaptionskip}
\makeatother

\renewcommand{\paragraph}{\roman{paragraph}}
\input cyracc.def

\newcommand{\Z}{\mathbb{Z}}
\newcommand{\F}{\mathbb{F}}

\newtheorem{thm}{\scshape   Theorem}[section]
\newtheorem{lem}[thm]{\scshape   Lemma}
\newtheorem{coro}[thm]{\scshape  Corollary}
\newtheorem{prop}[thm]{\scshape  Proposition}

\newtheorem{exa}[thm]{\scshape   Example}
\newtheorem{rem}[thm]{\scshape   Remark}

\begin{document}

\title{\bf On the structure of $1$-generator quasi-polycyclic codes over finite chain rings
\thanks{This research is supported by the National Natural Science Foundation of China (12071001, 61672036),
the Excellent Youth Foundation of Natural Science Foundation of Anhui Province (1808085J20), and the
Academic Fund for Outstanding Talents in Universities (gxbjZD03).
}
}
\author{
Rongsheng Wu\thanks{Key Laboratory of Intelligent Computing and Signal Processing, Ministry of Education, School of Mathematical Sciences,
Anhui University, Hefei, 230601, P.R. China, {\tt wrs2510@163.com}},
Minjia Shi\thanks{Corresponding author, Key Laboratory of Intelligent Computing and Signal Processing, Ministry of Education, School of Mathematical Sciences,
Anhui University, Hefei, 230601, P.R. China, {\tt smjwcl.good@163.com}},
and
Patrick Sol\'{e}\thanks{lab I2M (CNRS, Aix-Marseille University, Centrale Marseille), Marseilles, France, {\tt sole@enst.fr}}
}
\date{}
\maketitle
{\bf Abstract:} {Quasi-polycyclic (QP for short) codes over a finite chain ring $R$ are a generalization of quasi-cyclic codes, and
these codes can be viewed as an $R[x]$-submodule of $\mathcal{R}_m^{\ell}$, where $\mathcal{R}_m:= R[x]/\langle f\rangle$,
and $f$ is a monic polynomial of degree $m$ over $R$. If $f$ factors uniquely into monic and coprime basic irreducibles, then their algebraic structure allow us to characterize the generator polynomials and the minimal generating sets of 1-generator QP codes as $R$-modules. In addition, we also determine the parity check polynomials for these codes by using
the strong Gr\"{o}bner bases. In particular, via Magma system, some quaternary codes with new parameters are derived from these 1-generator QP codes.
}



\section{Introduction}

Codes over finite rings have attracted vital attention in the past several decades. In 1994s, Hammons \emph{et al.} \cite{HKCSS94} found that some
good nonlinear binary codes can be viewed as the Gray images of cyclic codes over $\mathbb{Z}_4$. After this seminal paper, codes over rings received lots of interest. It is worth noting that many follow-up works are related to finite chain rings.

Let $R$ be a finite commutative chain ring with maximal ideal $M$, and $\ell$ be a nonnegative integer.
We say that a linear code $C$ of length $n=m\ell$, with index $\ell$ and associate vector $(c_0,c_1,\ldots,c_{m-1})$, over $R$ is quasi-polycyclic if
it is held invariant under right multiplication by the matrix $\bar{D}={\rm diag}(\underbrace{D,D,\ldots,D}_\ell)$ (see \cite{ADLS16,LPS09}
for more details), where $c_0\in R\backslash M$ and
$D$ is the $m\times m$ matrix given by
$$D=\left(
  \begin{array}{ccccc}
    0 & 1 & 0 & \cdots & 0 \\
    0 & 0 & 1 & \cdots & 0 \\
    \vdots & \vdots & \vdots & \ddots & \vdots \\
    0 & 0 & 0 & \cdots & 1 \\
    c_0 & c_1 & c_2 & \cdots & c_{m-1} \\
  \end{array}
\right).$$

Denote $\mathcal{R}_m := R[x]/\langle f\rangle$, where $f=x^m-c(x)$ and $c(x)=c_0+c_1x+\cdots+c_{m-1}x^{m-1}\in R[x]$.
Using polynomial representation, it is easy to check that a QP code $C$ of length $n=m\ell$ with index $\ell$ over $R$ can be viewed as
an $R[x]$-submodule of $\mathcal{R}_m^{\ell}.$ Some variants of QP codes, are as follows:
\begin{itemize}
   \item When $f=x^n-1$, the code $C$ is called quasi-cyclic (QC), which is an important class of linear codes and enjoy a strong algebraic structure \cite{Cao13, CS93, GSF15, GS14, LS03(a)}. Furthermore, self-dual QC codes meet a modified GV bound \cite{LS03}, i.e., these codes are asymptotically good.
   \item It is well-known that $C$ is a polycyclic code if  $\ell=1.$ The algebraic structures of polycyclic codes over finite chain rings with a new duality were presented in \cite{ADLS16, FMB19}, and it is worthy to note that polycyclic codes over Galois rings and mixed alphabet $\Z_2\Z_4$ were studied in \cite{LOOS13} and \cite{WS21}, respectively;
   \item When $\ell=2,$ a $1$-generator QP code is the well-known double polycirculant (DP) code. The DP codes include special class
   of double circulant (DC) codes and double negacirculant codes \cite{B80,SQS18,SZ20}. In \cite{SSX18}, the isodual DP codes were investigated, and these codes were shown to be asymptotically good based on the existence of infinitely many irreducible trinomials over $\F_2$.
\end{itemize}

As we know, QP codes include DP codes as a subclass. If the length is not large enough, their parameters are optimal or
quasi-optimal amongst formally self-dual codes.
More importantly, in the binary case, it was shown that these codes are asymptotically good.
In fact, their relative distance satisfies the GV bound of rate one-half. Some DC codes of short lengths over finite fields are optimal \cite{Gui00}, and the Gray images of some special DC codes over finite rings satisfy a modified GV bound
 \cite{AOS18, HSS19, SHSS19(a)}.
The aim of this paper is to generalize the algebraic structure of quasi-cyclic codes to the QP codes, and the main tool used here
is the theory of strong Gr\"{o}bner basis.

Recently, some good linear (cyclic) codes over $\Z_4$ were obtained from constacyclic codes or generalized QC codes over different rings \cite{BBA17,GGWF17,OUA16,SQS17,WGF15}, and these codes have better parameters than the best know $\Z_4$-linear codes \cite{AA15}.
Moreover, in \cite{CCDL19}, new cyclic self-dual $\Z_4$-codes are obtained by employing a new representation for each cyclic code over $\Z_4$ of length $2n$.
In this paper, we also give some concrete examples of 1-generator QP codes over $\Z_4$ which lead to quaternary linear codes with new parameters.

The remainder of this paper is organized as follows. Section \ref{Sec:2} introduces the relevant terminology and notation. In Section \ref{Sec:3}, we describe the structural
properties of $1$-generator QP codes over finite chain rings, including their minimal generating set. The duals of QP
codes over the finite chain ring $R$ are described in Section \ref{Sec:4}. Section 5 is devoted to the construction of the new $\Z_4$-linear codes
from 1-generator QP codes with index two and three. Section 6 concludes the paper, along with some interesting problems.

\section{Preliminaries}\label{Sec:2}
\subsection{Finite chain rings}

It is well known that a finite commutative ring $R$ with identity $1\neq 0$ is called a finite chain ring if its ideals are linearly ordered by inclusion.
We start by introducing the properties for a finite chain ring. Readers interested should consult \cite{NS01}.

\begin{prop}
Let $R$ be a finite chain ring. Let $\gamma$ be a fixed generator of $M$ and $s$ the nilpotency index index of $\gamma$, that is, the smallest positive integer for which $\gamma^s=0$. Then
\begin{enumerate}
  \item [\rm (i)] the distinct proper ideals of $R$ are $\langle\gamma^t\rangle$, where $\langle\gamma^t\rangle=\gamma^tR$ for $1\leq t\leq s-1$;
  \item [\rm (ii)] for any element $e\in R\backslash\{0\}$, there is a unique $i$ and a unit $u\in R$ such that $e=u\gamma^i$ where $0\leq i\leq s-1$ and $u$ is unique modulo $\gamma^{s-i}$;
  \item [\rm (iii)] ${\rm Ann}_R(\gamma^i)=\langle\gamma^{s-i}\rangle$, where the annihilator ${\rm Ann}_R(a)$ of an element $a$ in $R$ is defined as ${\rm Ann}_R(a)=\{x\in R: xa=0\}$.
\end{enumerate}
\end{prop}

Let $F$ be the residue field of $R$, and the projection $R\rightarrow F$ extends naturally to a projection $R[x]\rightarrow F[x]$. Then for any $g(x)\in R[x]$ we denote by $\bar{g}(x)$ its image under this projection. In the following, a polynomial $g(x)\in \mathcal{R}_m$ or $R[x]$ will be denoted simply by $g$. It will simplify the notation to avoid writing out this each time.

\subsection{Polycyclic codes}

A polynomial over the field $F$ has no multiple irreducible factors is called square-free. Recall that a polynomial $g\in R[x]$ is called basic irreducible if $\bar{g}$ is irreducible over $F$.
From \cite[Lemma 3.1]{S06}, it is not difficult to see that the ring $R[x]/\langle f\rangle$ is principal if $\bar{f}$ is square-free. In this case, $f$ factors uniquely into monic and coprime basic irreducibles, this assumption we use throughout this paper.

To help characterize polycyclic codes over $R$, it is useful to introduce the content of the strong Gr\"{o}bner basis over a principal ideal ring according to
\cite[Theorem 4.2]{NS03} and the description below \cite[Theorem 4.3]{NS03}.

\begin{thm}\label{grobner}
Let $C \subseteq R[x]/\langle f\rangle$ be a nonzero polycyclic code. Then $C$ admits a set of generator polynomials $G$:
$$G=\{\gamma^{j_0}g_0,\gamma^{j_1}g_1,\ldots,\gamma^{j_t}g_t\},$$
where $0\leq t \leq s-1$ and
\begin{enumerate}
  \item [\rm (i)] $0\leq j_0<j_1<\ldots<j_t\leq s-1$;
  \item [\rm (ii)] $g_i$ is monic for $i=0,1,\ldots,t$;
  \item [\rm (iii)] $m>\deg(g_0)>\deg(g_1)>\ldots>\deg(g_t)$;
  \item [\rm (iv)] $g_t|g_{t-1}|\ldots|g_0|f$.
\end{enumerate}
\end{thm}

\begin{rem}\label{principal}
We follow the notation of Theorem \ref{grobner}, and the set $G$ is called a generating set in standard form for the polycyclic code $C$.
Then, the proof of \cite[Corollary to Theorem 6]{CS95} implies that $C=\langle\gamma^{j_0}g_0+\gamma^{j_1}g_1+\cdots+\gamma^{j_t}g_t\rangle$.

It should be noted that a strong Gr\"{o}bner basis of $C$ described in Theorem \ref{grobner} is not necessarily unique \cite[Theorem 7.5]{NS01}. However, the cardinality of the basis, the degree of their polynomials and the exponents $j_i$ for $0 \leq i\leq t$ are unique, and $C$ has $|F|^\Delta$ codewords, where $\Delta=\sum\limits_{i=0}^t(s-j_i)(\deg(g_{i-1})-\deg(g_i))$, $\deg(g_{-1})=m$, see \cite{FMB19} for more details.
\end{rem}

\begin{rem}
In the light of \cite[Theorem 4.4]{NS00}, let $C$ be a cyclic code of length $N$ with a unique generating set in standard form (see \cite[Definition 4.1]{NS00}) as $G=\{\gamma^{j_0}g_0,\gamma^{j_1} g_1,\ldots ,\gamma^{j_t}g_t\}$ over $R$. Consider the shortened cyclic code $C'$ of length $n$ obtained from $C$. Then
$$x^n=g_0q+r, \ \ \ {\rm where}\  \deg(r)<\deg(g_0).$$
Let $f=x^n-r$, and it is clear that $g_t|g_{t-1}|\cdots |g_0|f$. Note that the code $C'$ with generating set $G$ is an ideal of the quotient ring $R[x]/\langle f \rangle$, i.e., a polycyclic code of length $n$ over $R$. In other words, every shortened cyclic code over the finite chain ring is polycyclic.
\end{rem}

Let $C \subseteq R[x]/\langle f\rangle$ be a nonzero polycyclic code. Define
$${\rm Ann}_{\mathcal{R}_m}(C)=\{g\in R[x]/\langle f\rangle:gc=0\ {\rm for\ all} \ c\in C\}.$$

It is clear that ${\rm Ann}_{\mathcal{R}_m}(C)$ is a $R[x]$-submodule of $\mathcal{R}_m$. We consider now the problem of determining the standard generating set for ${\rm Ann}_{\mathcal{R}_m}(C)$.

\begin{prop}\label{ann}
Let $C=\langle g\rangle$ be a polycyclic code with the generating set $G=\{\gamma^{j_0}g_0,\gamma^{j_1}g_1,$ $\ldots, \gamma^{j_t}g_t\}$ over $R$, where $g=\gamma^{j_0}g_0+\gamma^{j_1}g_1+\cdots+\gamma^{j_t}g_t$. Then the set $H=\{\gamma^{b_0}h_0,\gamma^{b_1}h_1,\ldots,$ $\gamma^{b_{t+1}}h_{t+1}\}$ forms a strong Gr\"{o}bner basis of ${\rm Ann}_{\mathcal{R}_m}(C)$, where
\begin{enumerate}
  \item [\rm (i)] $b_i=s-j_{t+1-i}$ for all $i\in \{0,1,\ldots,t+1\}$. In particular, put $j_{t+1}=s$;
  \item [\rm (ii)] $h_i$ is monic with $h_i=f/g_{t-i}$ for $i\in\{0,1,\ldots,t+1\}$. In particular, put $g_{-1}=f$.
\end{enumerate}
\end{prop}
\begin{proof}
Consider an $\mathcal{R}_m$-module homomorphism $\theta$ from $\mathcal{R}_m$ onto $\langle g\rangle$ by:
$$a\longmapsto ag, \ {\rm for \ all}\ a\in \mathcal{R}_m.$$

For $0\leq k\leq t$ and $0\leq r\leq t+1$, if $j_k+b_r\geq s$, we have $\gamma^{j_k+b_r}g_kh_r=0$ since $\gamma^s=0$; if $j_k+b_r<s$, i.e.,
$j_k\leq j_{t+1-r}-1$ from condition (i), then we have $k\leq t-r$, so $g_{t-r}|g_k$. Thus, we have $$\gamma^{j_k}g_k\gamma^{b_r}h_r=\gamma^{j_k+b_r}g_k\frac{f}{g_{t-r}}=0,$$
which implies that $\langle h\rangle\subseteq \ker(\theta)$. On the other hand, by the first isomorphism theorem we have
$\mathcal{R}_m/\ker(\theta)\cong \langle g\rangle$. Then by a simple calculation from Remark \ref{principal} yields $|\langle h\rangle|=|\ker(\theta)|$.
Consequently, $\langle h\rangle=\ker(\theta)={\rm Ann}_{\mathcal{R}_m}(C)$.
\end{proof}

\subsection{Quasi-sequential codes}

Define linear functions $\sigma_i:R^m\rightarrow R$ for $1\leq i \leq \ell$. Let $C$ be a linear code of length $m\ell$ over $R$. Then $C$ is called quasi-sequential if for any $(\textbf{c}_1,\textbf{c}_2,\ldots,\textbf{c}_\ell)\in C$, we have
$(\textbf{c}'_1,\textbf{c}'_2,\ldots,\textbf{c}'_\ell)\in C$, where $\textbf{c}_i=(c_{i,0},c_{i,1}\ldots,c_{i,m-1})\in R^m$, and
$\textbf{c}'_i=(c_{i,1},\ldots,c_{i,m-1},\sigma_i(\textbf{c}_i))$ with $1\leq i \leq \ell$.

Let $C$ be a quasi-sequential code with associate vector $(\textbf{w}_1,\textbf{w}_2,\ldots,\textbf{w}_{\ell})$, where $\textbf{w}_i=(w_{i,0},w_{i,1},\ldots,w_{i,m-1})$ with $w_{i,0}\in R\backslash M$ and $1\leq i \leq \ell$. Similar to the argument in \cite{LPS09}, the code $C$ is invariant under right multiplication by the matrix $M={\rm diag}(M_1,M_2,\ldots,M_\ell)$, and for $1\leq i \leq \ell$, the $m\times m$ matrix $M_i$ is of the form:
$$M_i=\left(
  \begin{array}{ccccc}
    0 & 0 &  \cdots & 0 & w_{i,0} \\
    1 & 0 &  \cdots & 0 & w_{i,1} \\
     0 & 1 &  \cdots & 0 & w_{i,2} \\
    \vdots & \vdots &  \ddots & \vdots & \vdots \\
    0 & 0 &  \cdots & 1 & w_{i,m-1} \\
  \end{array}
\right).$$

It is easy to see that quasi-sequential codes are a natural generalization of quasi-cyclic codes and even of quasi-twisted codes. In fact, sequential codes, quasi-sequential codes with $\ell=1$, were first introduced by Hou \emph{et al.} \cite{HLP09} in 2009s, these codes are interesting for code constructions: they can be viewed as a source in obtaining good (even optimal) codes.

\section{Algebraic structure of $1$-generator QP codes}\label{Sec:3}

An $r$-generator QP code is a $R[x]$-submodule with $r$ generators. In this section, we only consider 1-generator QP codes over $R$.
Assume that $C$ is a $1$-generator QP code generated by $(F_1,F_2,\ldots,F_{\ell})\in \mathcal{R}_m^\ell$, then $C$ is of the form:
$$C=\{(gF_1,gF_2,\ldots,gF_{\ell}):g\in \mathcal{R}_m\}.$$
Then we define
$${\rm Ann}_{\mathcal{R}_m}(C)=\{g\in \mathcal{R}_m:gF_i=0\ {\rm for\ all} \ 1\leq i\leq \ell\}.$$
It is also clear that ${\rm Ann}_{\mathcal{R}_m}(C)$ is a $R[x]$-submodule of $\mathcal{R}_m$. According to Theorem \ref{grobner}, there exists a strong Gr\"{o}bner basis $\{\gamma^{k_0}h'_0,\gamma^{k_1}h'_1,\ldots,\gamma^{k_{r+1}}h'_{r+1}\}$ such that ${\rm Ann}_{\mathcal{R}_m}(C)=\langle h'\rangle$, where $h'=\gamma^{k_0}h'_0+\gamma^{k_1}h'_1+\cdots+\gamma^{k_{r+1}}h'_{r+1}$. We call $h'$ a parity check polynomial in standard form of the 1-generator QP code $C$.

For $1\leq i\leq \ell$, the projection of $\mathcal{R}_m^\ell$ on $\mathcal{R}_m$ according to its $i$-th component is the mapping $\phi_i:  \mathcal{R}_m^{\ell}\longrightarrow \mathcal{R}_m$ by $(a_1,a_2,\ldots,a_{\ell}) \longmapsto a_{i}$. Then we have the following proposition.

\begin{prop}\label{projection}
Let $C$ be a 1-generator QP code with the generator $\mathcal{G}=(F_1,F_2,\ldots,F_{\ell})\in \mathcal{R}_m^{\ell}$. Then $F_i\in C_i=\phi_i(C)$, where $C_i$ is a polycyclic code of length $m$ over $R$. Furthermore, ${\rm Ann}_{\mathcal{R}_m}(C)$ is a $R[x]$-submodule of ${\rm Ann}_{\mathcal{R}_m}(C_i)$ for $1\leq i\leq \ell$.
\end{prop}
\begin{proof}
Following the projection map defined above, it is easy to check that the code $\phi_i(C)$ is polycyclic over $R$. Therefore, the remainder of the argument follows from the fact that ${\rm Ann}_{\mathcal{R}_m}(C_i)$ is also a $R[x]$-submodule of $\mathcal{R}_m$.
\end{proof}

Next, we give a characterization of the parity check polynomial by using the strong Gr\"{o}bner basis with respect to the ideal generated by $F_1,F_2,\ldots,F_\ell$.
Suppose that
$\widehat{C}=\langle F_1,F_2,\ldots,F_\ell\rangle$
with the strong Gr\"{o}bner basis $\{\gamma^{a_0}f_0,\gamma^{a_1}f_1,\ldots,\gamma^{a_r}f_r\}$. Then the exact form of $h'$ is shown by the following result.

\begin{prop}\label{projection1}
Keep the above notation. Then $\{\gamma^{k_0}h'_0,\gamma^{k_1}h'_1,\ldots,\gamma^{k_{r+1}}h'_{r+1}\}$ forms a strong Gr\"{o}bner basis for
${\rm Ann}_{\mathcal{R}_m}(C)$, where
\begin{enumerate}
  \item [\rm (i)] $k_i=s-a_{r+1-i}$ for all $i\in \{0,1,\ldots,r+1\}$. In particular, put $a_{r+1}=s$;
  \item [\rm (ii)] $h'_i$ is monic with $h'_i=f/f_{r-i}$ for $i\in \{0,1,\ldots,r+1\}$. In particular, put $f_{-1}=f$.
\end{enumerate}
\end{prop}
\begin{proof}
The result follows from ${\rm Ann}_{\mathcal{R}_m}(C)={\rm Ann}_{\mathcal{R}_m}(\widehat{C})$ and Proposition \ref{ann}.
\end{proof}

Before the start of studying the minimal generating set of the QP code as $R$-module, we first introduce the following lemma.

\begin{lem}\label{polynomial}
Let $C$ be a 1-generator QP code over $R$ with generator $\mathcal{G}=(F_1,F_2,\ldots,F_{\ell})\in \mathcal{R}_m^{\ell}$. Then $F_i$ can be selected to be of the form $F_i=f_{i,0}+\gamma f_{i,1}+\cdots+\gamma^{s-1}f_{i,s-1}$, where $f_{i,j}$'s are monic or 0 in $R[x]$ with $1\leq i\leq \ell$ and
$0\leq j\leq s-1$.
\end{lem}
\begin{proof}
From Proposition \ref{projection}, it is easy to check that $F_i\in C_i$, where $C_i$ is a polycyclic code over $R$. By Theorem \ref{grobner}, there exist
polynomials $g_{i,0},g_{i,1},...,g_{i,t}\in R[x]$ such that
$$C_i=\langle \gamma^{j_{i,0}}g_{i,0}+\gamma^{j_{i,1}}g_{i,1}+\cdots +\gamma^{j_{i,t}}g_{i,t} \rangle.$$
Since $F_i\in C_i$, there exists a polynomial $f_i\in R[x]$ such that
$$F_i=f_i(\gamma^{j_{i,0}}g_{i,0}+\gamma^{j_{i,1}}g_{i,1}+\cdots +\gamma^{j_{i,t}}g_{i,t})=f_{i,0}+\gamma f_{i,1}+\cdots+\gamma^{s-1}f_{i,s-1},$$
where $f_{i,j}$'s are monic polynomials in $R[x]$ for all $1\leq i\leq \ell$ and $0\leq j\leq s-1$.
In particular, let $f_{i,j}=0$ (or $f$) if there is no term of $\gamma^j$ in $F_i$.
\end{proof}

\begin{exa}
Let $R=\Z_8,\ f=(x+3)(x^2+x+1)(x^4+2x^2+3x+1)$, and
\begin{eqnarray*}
    F_1 &=& 2x^3\big((x+3)(x^4+2x^2+3x+1)+2(x^4+2x^2+3x+1)+4\big) \\
     &\equiv& 2x^5+2x^4+4x^3+2x^2+2 {\pmod f},
  \end{eqnarray*}
where $f_1=2x^3$, $g_{1,0}=(x+3)(x^4+2x^2+3x+1)$, $g_{1,1}=x^4+2x^2+3x+1$, $g_{1,2}=1$. Then we can choose $f_{1,0}=0,\ f_{1,1}=x^5+x^4+x^2+1, \ f_{1,2}=x^3$ such that
$$F_1=f_{1,0}+2f_{1,1}+4f_{1,2},$$
where $f_{1,j}$ are monic or 0 with $0\leq j \leq 2$.
\end{exa}

For each $1\leq j\leq s-1$, let
$$h^{(j)}=\frac{f}{\gcd\big(f_{1,j}\prod\limits_{e=0}^{j-1}h^{(e)},f_{2,j}\prod\limits_{e=0}^{j-1}h^{(e)},\ldots,f_{\ell,j}\prod\limits_{e=0}^{j-1}h^{(e)},
f\big)},$$
where $\deg(h^{(j)})=r_j$. In particular, if $j=0$, let
$$h^{(0)}=\frac{f}{\gcd(f_{1,0},f_{2,0},\ldots,f_{\ell,0},f)},$$
where $\deg(h^{(0)})=r_0$. It should be noted that the expression of $f_{i,j}$ is not necessary unique, but the degree of $h^{(j)}$ is unique.
Then the following theorem gives the minimal generating set for a 1-generator QP code $C$ as an $R$-module.

\begin{thm}\label{genset}
Keep the above notation. Suppose that $\sum\limits_{e=0}^{s-1}\gamma^{e}f_{i,e}$ is not a zero divisor for $1\leq i\leq \ell$. Let
\begin{eqnarray*}
  G_0 &=& \{\mathcal{G},x\mathcal{G}, \ldots, x^{r_0-1}\mathcal{G}\},\\
  B_i &=& \big(\sum\limits_{e=i}^{s-1}\gamma^{e}f_{j,e}\prod\limits_{r=0}^{i-1}h^{(r)}\big)_{j=1}^{\ell},\ \ 1 \leq i\leq s-1, \\
  S_i &=& \{B_i, xB_i, \ldots,x^{r_i-1}B_i\},\ \  1 \leq i\leq s-1.
\end{eqnarray*}
Then $G_0\bigcup\limits _{i=1}^{s-1}S_i$ forms a minimal generating set for $C$ as an $R$-module. In addition, the 1-generator QP code $C$ has $|F|^{\sum\limits_{i=0}^{s-1}(s-i)r_i}$ codewords.
\end{thm}
\begin{proof}
Let $\textbf{c}=k\mathcal{G}$ be a codeword in $C$ with $k\in R[x]$. Recall the Division Algorithm, there exist two unique polynomials $p_1$ and $q_1$ such that $k=p_1h^{(0)}+q_1$, where $\deg{(q_1)}<r_0$ or $q_1=0$. Thus we have
\begin{eqnarray*}
  \textbf{c} = k\mathcal{G} &=& (p_1h^{(0)}+q_1)\mathcal{G}\\
   &=& p_1h^{(0)}\big(\sum_{i=0}^{s-1}\gamma^{i}f_{1,i},\ldots,\sum_{i=0}^{s-1}\gamma^{i}f_{\ell,i}\big)+q_1\mathcal{G}.
\end{eqnarray*}
Note that $q_1\mathcal{G}\in {\rm Span}(G_0)$. For $1\leq i \leq s-1$, there exist unique $p_{i+1}$ and $q_{i+1}$ such that $p_i=p_{i+1}h^{(i)}+q_{i+1}$ with $\deg(q_{i+1})<r_i$ from the Division Algorithm. Then we continue the steps above, it is easy to check that $\textbf{c}\in {\rm Span}(G_0)\bigcup _{i=1}^{s-1}{\rm Span}(S_i)$.

Therefore, it remains to show that ${\rm Span}(G_0),{\rm Span}(S_1),\ldots , {\rm and} \ {\rm Span}(S_{s-1})$ are pairwise disjoint. We just give a partial proof here, that is, ${\rm Span}(G_0)\cap {\rm Span}(S_1)=\{\textbf{0}\}$, and a similar argument for the rest of the cases.
Let $Q=(Q_1,Q_2,\ldots,Q_\ell)\in {\rm Span}(G_0)\cap {\rm Span}(S_1)$. Then there are $u,v \in R[x]$ such that
\begin{eqnarray}
  Q_i &=& \sum_{e=0}^{s-1}\gamma^{e}f_{i,e} u,\ \  {\rm where} \ \deg(u)<r_0,\ 1\leq i\leq \ell,\\
  Q_i &=& \sum_{e=1}^{s-1}\gamma^{e}f_{i,e}h^{(0)} v, \ \ {\rm where} \ \deg(v)<r_1,\ 1\leq i\leq \ell.
\end{eqnarray}
Therefore, Equs. (1) and (2) imply that $\sum\limits_{e=0}^{s-1}\gamma^{e}f_{i,e}(u-h^{(0)}v)=0$ due to $h^{(0)}f_{i,0}=0$. Since $\sum\limits_{e=0}^{s-1}\gamma^{e}f_{i,e}$ is not a zero divisor for $1\leq i\leq \ell$, we have $u=h^{(0)}v$.
On the other hand, it is easy to check the coefficients of $u$ belong to $\langle\gamma^{j_1-j_0}\rangle$ since $\gamma^{s-j_1}Q_i=0$ from Equ. (2).
The discussion above implies that $u=v=0$, which completes the proof.
\end{proof}

\section{Duality}\label{Sec:4}

Define the dual of the QP code $C$ of length $m\ell$ over $R$ by
$$C^{\bot}=\{\textbf{x}\in R^{m\ell}: \textbf{x}\cdot\textbf{y}=0\ {\rm for\  all} \ \textbf{y}\in C\}.$$
In general the linear code $C^{\bot}$ is not quasi-polycyclic. Then an alternate duality introduced in \cite{ADLS16, FMB19} on $\mathcal{R}_m$
is necessary. Let $g,h \in \mathcal{R}_m$ with degree less than $m$, define
$$\langle g,h\rangle_f=gh(0).$$
From now on, let $C$ be a QP code of length $n=m\ell$ with index $\ell$ over $R$. Then we define the annihilator dual $C^0$ of $C$ by the formula
$$C^0=\{g\in \mathcal{R}_m^\ell: \langle g,h\rangle_f=\sum_{i=1}^\ell \langle g_i,h_i\rangle_f=0 \ \ {\rm for \  all} \ \ h\in C\},$$
where $g=(g_1,g_2,\ldots,g_\ell)$ and $h=(h_1,h_2,\ldots,h_\ell)\in \mathcal{R}_m^\ell$.

\begin{prop}\label{dual}
Let $C$ be a $R[x]$-submodule of $\mathcal{R}_m^\ell$, that is, a QP code of length $n=m\ell$ with index $\ell$ over $R$. Then $C^0$ is also a $R[x]$-submodule of $\mathcal{R}_m^\ell$.
\end{prop}
\begin{proof}
Let $\lambda\in R[x]$ with degree less than $m$. Suppose that $u\in C$ and $v\in C^0$. Notice that $\langle u,v\rangle_f=0$, we have
$$\langle u,\lambda v\rangle_f=\lambda(0)\langle u,v\rangle_f=0.$$
This implies that $\lambda v\in C^0$. Thus, $C^0$ is a $R[x]$-submodule of $\mathcal{R}_m^\ell$, i.e., a QP code of length $m\ell$ with index $\ell$ over $R$.
\end{proof}

\begin{rem}
Let $f=x^m-c(x)$ with $c(x)=c_0+c_1x+\cdots +c_{m-1}x^{m-1}\in R[x]$ and $c_0\in R\backslash M$. It is easy to check that the inner product $\langle.,.\rangle_f$ on $\mathcal{R}_m^\ell$ is a nondegenerate bilinear form from a similar argument in \cite[Lemma 3.1]{ADLS16} (see also \cite[Lemma 3]{FMB19}). Moreover, if $\ell=1$, i.e., $C$ is a polycyclic code over $R$, then we have $C^0={\rm Ann}_{\mathcal{R}_m}(C)$.
\end{rem}

However, since the dual of a QP code is not quasi-polycyclic in general, there is a mathematical interest of the connection between the quasi-polycyclic code $C$ and its dual which is quasi-sequential.

\begin{thm}
Let $C$ be a linear code over $R$. Then $C$ is quasi-polycyclic if and only if $C^\bot$ is quasi-sequential with the same associate vector in each component.
\end{thm}
\begin{proof}
Let $H$ be the parity check matrix of $C$, this yields $CH^T=0$. If $C$ is quasi-polycyclic, we can choose $M_i=D^T$ for $1\leq i \leq \ell$, that is, $M=\bar{D}^T$. Then we have $$C\bar{D}H^T=C(H\bar{D}^T)^T=C(HM)^T=0.$$
Therefore, the dual code $C^\bot$ is invariant under right multiplication by the matrix $M$.
\end{proof}

\section{New $\mathbb{Z}_4$-linear codes from $1$-generator QP codes}\label{Sec:5}

In this section, we present some examples of $1$-generator QP codes with indexes 2 and 3 over $\Z_4$ by using the computer algebra system Magma \cite{BCC97}.
These quaternary codes have better parameters than the ones in \cite{AA15}. For this purpose, we limit our search scope to the case of free codes. The
following corollary gives a characterization for a quaternary $1$-generator QP code with index $\ell$ to be free.

\begin{coro}\label{free}
Let $C$ be a $1$-generator QP code with the generator $\mathcal{G}=(f_{1,0},f_{2,0},\ldots,f_{\ell,0})$ over $\Z_4$.
If $f_{i,0}$ is not a zero divisor for all $1\leq i\leq \ell$, then $C$ is a free QP code with $4^r$ codewords, where $r=\deg(h^{(0)})$, and
$$h^{(0)}=\frac{f}{\gcd(f_{1,0},f_{2,0},\ldots,f_{\ell,0},f)}.$$
\end{coro}
\begin{proof}
According to Theorem \ref{genset}, if $f_{i,0}$ is not a zero divisor for all $1\leq i\leq \ell$, then $h^{(0)}=f$ and $B_1=\textbf{0}$. Therefore, the code $C$ is a free quaternary QP code of size $4^r$ with the minimal generating set $\{\mathcal{G},x\mathcal{G}, \ldots, x^{r-1}\mathcal{G}\}$.
\end{proof}

We are now in a position to present some computational examples of new $\Z_4$-linear codes from 1-generator QP codes with indexes two and three in Table 1 and Table 2, respectively.

For convenience, we write the coefficients of the generator polynomials of QP codes in decreasing order. For example, let $10^21^3321$ denote the polynomial
$x^8+x^5+x^4+x^3+3x^2+2x+1$. The linear codes in Tables 1 and 2 all have better parameters than the best known $\Z_4$-linear codes \cite{AA15}, i.e., these codes have larger minimum Lee distance than the ones in \cite{AA15} with given length and code size. In particular, there are some linear codes in the tables that have the same minimum Lee distance with the best known nonlinear codes. For example, let $f=102013$ in Table 1 and $f=12^201^22$ in Table 2, then the linear codes with the parameters $(10,4^2,10)$ and $(18,4^2,18)$ are the same as the best known quaternary nonlinear codes \cite{AA15}.

On the other hand, in Table 1, the symbol $``*"$ denotes the obtained linear codes which have bigger code size than the best known $\Z_4$-linear codes \cite{AA15} with given code length and minimum Lee distance. For example, in Table 1, for $f=12^201^22$ and $1210210213$, the codes of length 12 and 18 with minimum Lee distances both 9, have sizes $4^3$ and $4^6$, respectively, while the codes in \cite{AA15} only have sizes $4^2$ and $4^32$, respectively.

Note that Table 1 and Table 2 contain only very small values of the length $n$. One can continue to construct more interesting linear codes from Theorem \ref{genset} and Corollary \ref{free} even for non-free case using Magma \cite{BCC97}.

\begin{table}
  \centering
  \caption{\newline Some new quaternary linear codes constructed by 1-generator QP codes with \\ {index two.}}
  \begin{tabular*}{\hsize}{@{}@{\extracolsep{\fill}}llll@{}}
  \toprule
  $f$ & $f_{1,0}$ & $f_{2,0}$  & \bf{Parameters}  \\
  \midrule
  $102013$ & $101^23$ & $130102$  & $(10,4^2,10)$  \\
  $102013$ & $12^21$ & $1^232^2$  & $(10,4^3,7)$  \\
  $12^201^22$ & $101^23$ & $13^22^23$  & $(12,4^3,9)^*$  \\
  $102^23231^2$ & $1302102$ & $1^20^212$  & $(16,4^4,11)$  \\
  $1013^2213^21$ & $10323^32$ & $12101^232^2$  & $(18,4^3,14)$  \\
  $1210210213$ & $1^221^23$ & $1^2321$  & $(18,4^6,9)^*$  \\
  $12310^23123231$ & $120^21^2032$ & $1020130102$  & $(24,4^5,16)$  \\
  $1231^20210131^2$ & $1313^2021$ & $1202^2101^2$  & $(24,4^6,14)$ \\
  $1^32030^33121$ & $12^23^22$ & $10^2312^2$  & $(24,4^8,12)$  \\
  $120323131201^2$ & $10^33$ & $13^212^2$  & $(24,4^9,10)$  \\
  $1021^2032^31213$ & $12^23^22$ & $10^2312^2$  & $(26,4^9,11)$  \\
  $102^213^210312^231$ & $1^20^231^223012$ & $13023^21212302$  & $(28,4^4,22)$ \\
  \bottomrule
\end{tabular*}
\end{table}

\begin{table}
  \centering
  \caption{\newline Some new quaternary linear codes constructed by 1-generator QP codes with \\{index three.}}
  \begin{tabular*}{\hsize}{@{}@{\extracolsep{\fill}}lllll@{}}
  \toprule
  $f$ & $f_{1,0}$ & $f_{2,0}$ & $f_{3,0}$  & \bf{Parameters}  \\
  \midrule
  $102013$ & $12^21$ & $13^22$ & $1^232^2$ & $(15,4^3,12)$  \\
  $12^201^22$ & $121312$ & $1^30232$ & $1^2032$ & $(18,4^2,18)$  \\
  $12^201^22$ & $101^23$ & $13^22^23$ & $1323$ & $(18,4^3,15)$  \\
  $12^232^23^31$ & $13132132^2$ & $1^332312$ & $12023^201$ & $(27,4^3,24)$  \\
  $1^2321023^3$ & $121^42^2$ & $103^21^22$ & $12313^20$ & $(27,4^4,20)$  \\
  $10310^21301$ & $1302102$ & $1^20^212$ & $132010$ & $(27,4^5,17)$  \\
  $1^2301^223^21^2$ & $12302302$ & $1012032$ & $1^23^2031$ & $(30,4^5,21)$  \\
  $10201201231$ & $10^2312^2$ & $12^23^22$ & $1^22101$ & $(30,4^6,18)$  \\
  $1^22310102^21$ & $130102$ & $1^2032$ & $1231^2$ & $(30,4^7,16)$  \\
  $1021^23^21^323$ & $12101^232^2$ & $10323^32$ & $123031^2$ & $(33,4^5,22)$  \\
  $123^31^221213$ & $10232102$ & $1203032$ & $1^2031^23$ & $(33,4^6,20)$  \\
  $12^23^31^32^23$ & $10^2312^2$ & $12^23^22$ & $102310$ & $(33,4^7,18)$  \\
  $1231^20210131^2$ & $1313^2021$ & $132120121$ & $1202^2101^2$ & $(36,4^6,24)$   \\
  $1231^230^2101323$ & $13^20^2301^32$ & $1^2302323^212^2$ & $1^3203213^20$ & $(39,4^4,29)$   \\
  $1210^23012023^3$ & $12^23^210^212$ & $10^231^202102$ & $1023132010$ & $(39,4^5,27)$   \\
  \bottomrule
\end{tabular*}
\end{table}

\section{Conclusion}

In this paper, we generalize some previous studies on 1-generator quasi-cyclic codes over the finite fields to finite chain ring.
Using the Gr\"{o}bner basis theory, we determine their minimal generating sets and their parity check polynomials. Moreover, we have constructed
some $\Z_4$-linear codes with parameters better than the best known $\Z_4$-linear codes \cite{AA15}. Unlike the quasi-cyclic codes,
for a positive integer $m$, we can choose different $f\in R[x]$ with $\deg(f)=m$ and $\bar{f}$ square free, not only $f=x^m-1.$
Thus this extra flexibility is a suitable approach for constructing codes with good parameters over finite chain rings by using the Gray map or codes over finite fields.

In the future work, the problem of characterization of the self-annihilator quasi-polycyclic codes over $R$ with the inner product defined in Section 4 is open, and it might be easier to start with the $1$-generator case. Another topic is to see if the results in this paper can be extended to more general cases, for example, the quasi-polycyclic codes with $r \ (\geq2)$ generators.\\

{\bf Acknowledgement:} The authors would like to thank the anonymous referees for their careful checking and helpful suggestions which have improved the manuscript. The authors are also grateful to the Assoc. Prof. Edgar Mart\'{\i}nez-Moro for helpful discussions.

\end{document}